\documentclass[letterpaper,12pt]{article}
\usepackage{mathptmx}
\usepackage{graphics}

\usepackage[authoryear,comma,sectionbib]{natbib}

\usepackage{times}
\usepackage{bm}

\usepackage{mathrsfs}

\usepackage{graphicx}           
\usepackage{psfrag}
\usepackage{color}

\usepackage{amsbsy}
\usepackage{amssymb}
\usepackage{amsmath}
\usepackage{setspace}

\newtheorem{theorem}{Theorem}

\newtheorem{proof}{Proof}
\newtheorem{example}{Example}

\def\BE{\mbox{E}}             

\def\BP{\mbox{pr}}             

\def\d{\mathrm{d}}
\setcitestyle{comma}

\begin{document}
\doublespacing

\title{Approximate Bayesian computation (ABC) gives  exact results under the  assumption of model error}
\author{Richard D. Wilkinson$^*$\\
 School of Mathematical Sciences, University of Nottingham. NG7 2RD\\
r.d.wilkinson@nottingham.ac.uk}

\markboth{Richard D. Wilkinson}{Approximate Bayesian computation (ABC) gives  exact results under the   assumption of model error}

\maketitle

\begin{abstract}

Approximate Bayesian computation (ABC) or likelihood-free inference algorithms are used to find approximations to posterior distributions without making explicit use of the likelihood function, depending instead on simulation of sample data sets from the model. In this paper we show that under the assumption of the existence of a uniform additive model error term, ABC algorithms give exact results when sufficient summaries are used.  This interpretation  allows the approximation made in many previous application papers to be understood, and should guide the choice of metric and  tolerance  in future work.  ABC algorithms can be generalized by replacing the 0-1 cut-off with an acceptance probability that varies with the distance of the simulated data from the observed data. The acceptance density  gives the distribution  of  the error term, enabling the uniform error usually used to be replaced by a general distribution.
This generalization can also be applied to  approximate Markov chain Monte Carlo algorithms.
In light of this work, ABC algorithms can be seen as  calibration techniques for implicit stochastic models, inferring parameter values in light of the computer model, data, prior beliefs about the parameter values, and any measurement or model errors. 

\end{abstract}


\section{Introduction}

Approximate Bayesian computation (ABC) algorithms are a group of methods for performing Bayesian inference without the need for explicit evaluation of the model likelihood function \citep{Beaumont_etal02, Marjoram_etal03, Sisson_etal07}. The algorithms can be used with implicit computer models \citep{Diggle_etal84} that generate sample data sets rather than return likelihoods. ABC methods have become popular in the biological sciences \citep{Sunnaker_etal2013} with applications  in genetics  \citep[see, for example, ][]{Siegmund_etal08, Foll_etal08}, epidemiology \citep{Blum_etal10, Tanaka_etal06} and population biology \citep{Ratmann_etal07, Hamilton_etal05, Cornuet_etal08} most common. 
This popularity is primarily due to the fact that the likelihood function, which can be difficult or impossible to compute for some models, is not needed in order to do inference. However, despite their popularity little is known about the quality of the approximation they provide beyond results shown in simulation studies.

In this paper we give a framework in which the accuracy of ABC methods can be understood. The notation throughout this paper is as follows.
Let $\theta$ denote the vector of unknown model parameters we wish to infer, and let $\eta(\cdot)$ denote the computer simulator. We assume $\eta(\cdot)$ is stochastic, so that the  simulator repeatedly run at $\theta$ will give a range of possible  outputs, and write $X \sim \eta(\theta)$ to denote that $X$ has the same distribution as the model run at $\theta$. To distinguish the model output from the observed data, let $D$ denote the observations. The aim is to calibrate the model to the data, in order to learn about the true value of the parameter. The Bayesian approach is to find the posterior distribution of $\theta$ given $D$, given by
\begin{equation*}
\pi(\theta\mid D) =\frac{\pi(D \mid \theta) \pi(\theta)}{\pi(D)}.
 \end{equation*}
Throughout this paper, $\pi(\cdot)$ is used to denote different probability densities, and $\pi(\cdot\mid \cdot)$  conditional densities, with the context   clear from the arguments. Above, $\pi(\theta)$ is the prior distribution, $\pi(D\mid \theta)$ is the likelihood of the data under the model given parameter $\theta$  (the probability distribution of $\eta(\theta)$), $\pi(\theta\mid D)$ is the  posterior distribution, and $\pi(D)$ is the evidence for the model.

It is usual in Bayesian inference to find that the normalizing constant $\pi(D)$ is intractable, and a wide range of Monte Carlo techniques have been developed to deal with this case \citep{Liu_01}. Doubly-intractable distributions are distributions which have a likelihood function $\pi(D\mid \theta)=q(D\mid \theta)/c(\theta)$ which is known only up to a normalizing constant, $c(\theta)$, which is intractable. Standard Monte Carlo techniques do not apply to  these distributions, and \citet{Murray_etal06} have developed algorithms which can be used in this case. ABC methods are Monte Carlo techniques developed for use with completely-intractable distributions, where the likelihood function $\pi(D \mid \theta)$ is not even known up to a normalizing constant. ABC algorithms, sometimes called likelihood-free algorithms, enable inference using only simulations generated from the model, and do not require any evaluation of the likelihood. The most basic form of the ABC algorithm is based on the rejection algorithm, and is as follows:
\begin{enumerate}
\item[]{\bf Algorithm A: approximate rejection algorithm}
\item[A1] Draw $\theta\sim \pi(\theta)$
\item[A2] Simulate $X$ from the simulator $X\sim \eta(\theta)$
\item[A3] Accept $\theta$ if $\rho(X,D)\leq \delta$.
\end{enumerate}
Here, $\rho(\cdot, \cdot)$ is a distance measure on the model output space, and $\delta$ is a tolerance determining the accuracy of the algorithm. Accepted values of $\theta$ are not from the true posterior distribution, but from an approximation to it, written $\pi(\theta\mid \rho(D,X)\leq \delta)$. When $\delta=0$ this algorithm is exact and gives draws from the posterior distribution $\pi(\theta\mid D)$, whereas as $\delta \rightarrow \infty$ the algorithm gives draws from the prior. While  smaller values of $\delta$ lead to  samples which better approximate the true posterior, they also lead to lower acceptance rates in step A3 than larger values, and so more computation must be done to get a given  sample size. Consequently, the tolerance $\delta$ can be considered as controlling a trade-off between computability and accuracy.

Several  extensions  have been made to the approximate rejection algorithm. If the data are high dimensional, then a standard change to the  algorithm is  to summarize the model output and data, using a summary statistic $S(\cdot)$ to project $X$ and $D$ onto a lower dimensional space. Algorithm A is then changed so that step A3 reads
\begin{enumerate}
 \item[A$3^\prime$]  Accept $\theta$ if $\rho(S(X),S(D))\leq \delta$.
\end{enumerate}
 Ideally, $S(\cdot)$ should be chosen to be a sufficient statistic for $\theta$. However, if the likelihood is unknown, then sufficient statistics cannot be identified. Summarizing the data and model output using a non-sufficient summary adds another layer of approximation  on top of that added by the use of the distance measure and tolerance, but again, it is not known what effect any given choice for $S(\cdot)$ has on the approximation.


In this paper it is shown that the basic approximate rejection algorithm  can be interpreted as performing exact inference in the presence of uniform model or measurement error. 
In other words, it is shown that ABC  gives exact inference  for the wrong model, and  we give a distribution for the model error term for whatever choice of metric and tolerance  are used. This interpretation allows us to show the effect a given choice of metric, tolerance and weighting have had in previous applications, and should provide guidance when choosing metrics and weightings in future work.
It is also shown that Algorithm A can be generalized to give inference under the assumption of a completely flexible form for the model error. We discuss how  to model the model error, and show how some models can be rewritten to give exact inference. \citet{Ratmann_etal07} explored a related idea, and looked at using ABC algorithms to diagnose model inadequacies. The aim here is not  to diagnose errors, but to account for them in the inference so as to provide posteriors that take known inadequacies into account, and to understand the effect of using standard ABC approaches.

Finally, ABC has been extended by \citet{Marjoram_etal03} from the rejection algorithm to approximate Markov chain Monte Carlo algorithms,  and by \citet{Sisson_etal07}, \citet{Toni_etal09},   and \citet{Beaumont_etal09} to approximate sequential Monte Carlo algorithms. We extend the  approximate Markov chain Monte Carlo algorithm to give inference for a general form of error, and   suggest methods for calculating Bayes factors and integrals for completely-intractable distributions.

\section{Interpreting ABC}

In this section a framework is described which enables  the effect a given metric and weighting  have in ABC  algorithms to be understood. This will then allow us to improve the inference by carefully choosing a metric and weighting which more closely represents our true beliefs. The key idea is to assume that there is a discrepancy  between the best possible model prediction and the data. This discrepancy  represents either measurement error on the data, or  model error describing our statistical beliefs about where the model is wrong. George Box famously wrote that  `all models are wrong, but some are useful', and in order to  link models to reality it is necessary to account for this model error when performing inference. In the context of deterministic models, this practice is well established \citep{Campbell_06, Goldstein_etal08,  Higdon_etal08}, and  should also be undertaken when linking stochastic models to reality, despite the fact that the variability in the model can seemingly explain the data as they are.

The framework introduced here uses the best input approach, similar to that given in \citet{Kennedy_etal01}. We assume that the measurement $D$ can be considered as a realization of the model run at its best input value, $\hat{\theta}$, plus an independent error term $\epsilon$
\begin{equation}\label{eqn:disc_model}
D=\eta(\hat{\theta})+\epsilon.
\end{equation}
The error $\epsilon$ might represent measurement error on $D$, or model error in $\eta(\cdot)$, or both, in which case we write $\epsilon=\epsilon_1+\epsilon_2$. Discussion about the validity of Equation~(\ref{eqn:disc_model}), and what $\epsilon$ represents and how to model it are delayed until Section 3, and for the time being we simply consider $\epsilon$ to have density $\pi_\epsilon(\cdot)$. 
The aim is to describe our posterior beliefs about the best input $\hat{\theta}$ in light of the error $\epsilon$, the data $D$, and prior beliefs about $\hat{\theta}$. Consider the following algorithm:

\begin{enumerate}
\item[]{\bf Algorithm B: probabilistic approximate rejection algorithm}
\item[B1] Draw $\theta\sim \pi(\theta)$
\item[B2] Simulate $X$ from the model $X\sim \eta(\theta)$
\item[B3] Accept $\theta$ with probability $\frac{\pi_\epsilon(D-X)}{c}$.
\end{enumerate}
Here, $c$ is a constant chosen to guarantee that $\pi_\epsilon(D-X)/c$ defines a probability. For most cases we will expect $\epsilon$ to have a modal value of 0, and so taking $c=\pi_\epsilon(0)$ will make the algorithm valid and also ensure efficiency by maximizing the acceptance rate. If summaries are involved, or if $D$ and $X$ live in non-comparable spaces, so that $D-X$ does not make sense, we can instead use any distribution relating $X$ to $D$, $\pi_\epsilon(D|X)$ instead.


The main innovation in this paper is to show that Algorithm B gives exact inference for the statistical model described above by Equation~(\ref{eqn:disc_model}). This is essentially saying that ABC gives exact inference, but for the wrong model.

\begin{theorem}\label{thm:1}
Algorithm B gives draws from the posterior distribution $\pi(\hat{\theta}\mid D)$ under the assumption that $D=\eta(\hat{\theta})+\epsilon$ and $\epsilon \sim \pi_\epsilon(\cdot)$ independently of $\eta(\hat{\theta})$.
\end{theorem}

\begin{proof}
Let 
\begin{equation*}
I=\begin{cases} &1 \mbox{ if } \theta \mbox{ is accepted}\\
&0 \mbox{ otherwise.}\end{cases}
 \end{equation*}
We then find that 
\begin{align*}
\BP(I=1\mid \theta)&=\int \BP(I=1\mid \eta(\theta)=x, \theta)\pi(x\mid \theta)\d x\\
&=\int \frac{\pi_\epsilon(D-x)}{c} \pi(x\mid \theta) \d x.
\end{align*}
This gives that the distribution of accepted  values of $\theta$ is
\begin{equation*}
\pi(\theta\mid I=1)=\frac{\pi(\theta) \int \pi_\epsilon(D-x) \pi(x\mid \theta) \d x}{\int \pi(\theta') \int \pi_\epsilon(D-x) \pi(x\mid \theta') \d x \d \theta'}.
\end{equation*}
To complete the proof we must find the posterior distribution of the best model input $\hat{\theta}$ given the data $D$ under the assumption of  model error. Note that $\pi(D\mid \eta(\hat{\theta})=x)=\pi_\epsilon(D-x)$  which implies that the likelihood of $\theta$ is
\begin{align*}
\pi(D\mid \hat{\theta}) &=\int \pi(D\mid \eta(\hat{\theta})=x, \hat{\theta }) \pi(x\mid \hat{\theta}) \d x\\
&=\int \pi_\epsilon(D-x) \pi(x\mid \hat{\theta}) \d x.
\end{align*}
Consequently, the posterior distribution of $\hat{\theta}$ is
\begin{equation*}
\pi(\hat{\theta}\mid D)=\frac{\pi(\hat{\theta}) \int \pi_\epsilon(D-x) \pi(x\mid \hat{\theta}) \d x}{\int \pi(\theta) \int \pi_\epsilon(D-x) \pi(x\mid \theta) \d x \d \theta}
\end{equation*}
which matches the distribution of accepted values from Algorithm B.
\end{proof}

To illustrate the algorithm, we consider the toy example used in \citet{Sisson_etal07} and again in \citet{Beaumont_etal09} where analytic expressions can be calculated for the approximations.

\begin{example}

Assume the model is a mixture of two normal distributions with a uniform prior for the mean:
\begin{equation*}
\eta(\theta)\sim \frac{1}{2} \mathcal{N}(\theta,1)+\frac{1}{2} \mathcal{N}(\theta,\frac{1}{100}), \; \theta\sim \mathcal{U}[-10,10].
\end{equation*}
Further assume that we observe $D=0$, but that there is measurement error $\epsilon$ on this data. If $\epsilon \sim \mathcal{U}[-\delta, \delta]$, which is the assumption made when using Algorithm A with $\rho(x, 0)=| x| $, then it is possible to show that  the approximation is
\[
\pi(\theta\mid \epsilon \sim \mathcal{U}[-\delta, \delta], D=0)\propto \Phi\left(\delta-\theta\right)-\Phi(-\delta-\theta)+\Phi(10(\delta-\theta))-\Phi(-10(\delta+\theta))
\]
for $\theta\in[-10,10]$, where $\Phi(\cdot)$ is the standard Gaussian cumulative  distribution function. An
alternative to assuming uniform error, is to suppose that the error has a normal distribution $\epsilon\sim \mathcal{N}(0, \delta^2/3)$. It can then be shown that the posterior distribution of $\theta$ is
\[
\pi(\theta\mid \epsilon\sim \mathcal{N}(0, \frac{\delta^2}{3}), D=0)\propto \frac{1}{2}\phi(\theta; 0, 1+\frac{\delta^2}{3})+ \frac{1}{2}\phi(\theta; 0, \frac{1}{100}+\frac{\delta^2}{3})
\]
truncated onto $[-10, 10]$. This is the approximation found when using Algorithm B with a Gaussian acceptance kernel, where $\phi(\cdot; \mu, \sigma^2)$ is the probability density function of a Gaussian distribution with mean $\mu$ and variance $\sigma^2$. The value of the variance, $\delta^2/3$, is chosen to be equal to the variance of a $\mathcal{U}[-\delta, \delta]$ random variable.   For large values of the tolerance $\delta$, the difference between the two approximations  can be significant (see Figure~\ref{fig:toy_example}), but in the limit as $\delta$ tends to zero, the two approximations will be the same, corresponding to zero error.

\end{example}

\begin{figure}
\psfrag{y}[][][0.75]{Density}
\psfrag{x}[][][0.75]{$\theta$}
\psfrag{delta 0.1}[][][0.75]{$\delta=0.1$}
\psfrag{delta 1}[][][0.75]{$\delta=1$}
\includegraphics[width=\textwidth]{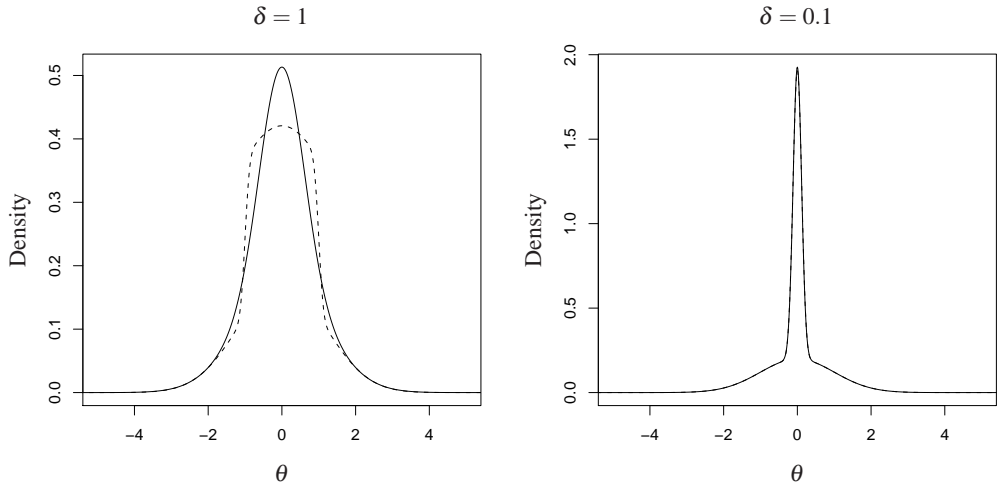}
\caption{The posterior distributions found when using Algorithm A (solid line) and Algorithm B (dashed line) with a Gaussian acceptance kernel. The left plot is for $\delta=1$ and the right plot for $\delta=0.1$.  The two curves are indistinguishable in the second plot. }
\label{fig:toy_example}
\end{figure}

\section{Model discrepancy}

The interpretation of ABC given  by Theorem 1 allows us to revisit previous analyses done using the  ABC algorithm, and to understand the approximation in the posterior in terms of the distribution  implicitly assumed  for the error term. If the approximate rejection algorithm (Algorithm A)  was used to do the analysis, we can see that this is equivalent to using the acceptance probability 
\[
\frac{\pi_\epsilon(r)}{c}=\begin{cases} 1 \mbox{ if } \rho(r) \leq \delta\\ 0 \mbox{ otherwise}\end{cases}
\]
 where $r$ is the distance between the simulated and observed data.
This says that Algorithm A gives exact inference for the model which assumes a uniform measurement error on the region defined by the 0-1 cut-off,  i.e., 
\[
\epsilon \sim \mathcal{U}\{x: \rho(x,D)\leq \delta\}.
\]
If $\rho(\cdot, \cdot)$ is a Euclidean metric, $\rho(D,x)=(x-D)^T(x-D)$, this is equivalent to assuming uniform measurement error on a ball of radius $\delta$ about $D$. 
In most situations, it is likely to be a poor choice for a model of the measurement error, because the tails of the distribution are short, with zero mass outside of the interval $[-\delta, \delta]$.

There are two ways we can choose to view the  error term; either as measurement error or model error.
Interpreting $\epsilon$ to represent measurement error is relatively straight forward, as scientists usually hold beliefs about the distribution and magnitude of measurement error on their data. For 
 most problems, assumptions of uniform measurement error will be inappropriate, and so  using Algorithm A with a 0-1 cut-off will be inappropriate.  But we have shown how to replace this uniform assumption with a distribution which  more closely represents the beliefs of the scientist. Although the distribution of the measurement error will often be  completely specified by the scientist,  for example zero-mean Gaussian error with known variance, it is possible to include unknown parameters for the distribution of $\epsilon$ in $\theta$ and infer them along with the model parameters.  Care needs to be taken to choose the constant $c$ so that the acceptance rate in step B3 is less than one for all values of the parameter, but other than this it is in theory simple to infer error parameters along with the model parameters. So for example, if $\epsilon \sim \mathcal{N}(0, \sigma^2)$, where $\sigma^2$ is unknown, we could include  $\sigma^2$ in $\theta$. 
  
Some models  have sampling or measurement error  built into the computer code so that the model output includes a realization of this noise. Rather than coding the noise process into the model, it will sometimes be possible to rewrite the model so that it outputs the latent underlying signal. If the likelihood of the data given the latent signal is computable (as it often is), then it may be possible to analytically account for the noise with the acceptance probability $\pi_\epsilon(\cdot)$. ABC methods have proven most popular in fields such as genetics, epidemiology, and population biology, where a common occurrence is to have  data generated by sampling  a hidden underlying tree structure. In many cases, it is the partially observed tree structure which causes the likelihood to be intractable, and given the underlying tree the sampling process will have a known distribution. If this is the case (and if computational constraints allow), 
 we can use the probabilistic ABC algorithm to do the  sampling to give exact inference without any assumption of model error. Note that if the sampling process gives continuous data, then exact inference using the rejection algorithm would not be  possible, and so this approach has the potential to give a significant improvement over current methods.

\begin{example}\label{example:BP}
To illustrate the idea of rewriting the model in order to do analytic sampling, we describe a  version of the problem considered in \citet{Plagnol_etal04} and \citet{Wilkinson_etal09}. Their aim was to use the primate fossil record to date the divergence time of the primates. They used an inhomogeneous branching process to model speciation, with trees rooted at time $t=\tau$, and simulated forwards in time to time $t=0$, so that the depth of the tree, $\tau$, represents the divergence time of interest. The branching process is parametrized by $\lambda$, which can either be estimated and fixed, or treated as unknown and given a prior distribution. Time is split into geologic epochs $\tau<t_{k}<\cdots<t_1<0$, and the data consist of counts of the number  of primate species that have been found in each epoch of the  fossil record, $D=(D_1, \ldots, D_{k})$. Fossil finds are modelled by a discrete marking process on the tree, with each species having equal probability $\alpha$ of being preserved as a fossil in the record.
If we let $N_i$ be the cumulative number of branches that exist during any point of epoch $i$, then the model used for the fossil finds process can be written as
$ D_i \sim \mbox{Binomial}(N_i, \alpha)$.
The distribution of $N=(N_1, \ldots, N_{14})$  cannot be calculated  explicitly and so we cannot use a likelihood based approach to find the posterior distribution of the unknown parameter $\theta=(\lambda, \tau, \alpha)$. The  ABC approach used in \citet{Plagnol_etal04} was to draw a  value of $\theta$ from its prior, simulate a sample tree and fossil finds, and then count the number of simulated fossils in each epoch to find a simulated value of the data $X$. They then accepted $\theta$ if $\rho(D, X)\leq \delta$ for some metric $\rho(\cdot, \cdot)$ and tolerance $\delta$. This gives an approximation to the posterior distribution of the parameter given the data and the model, where the approximation can be viewed as model or measurement error.

However, instead of approximating the posterior, it is possible in theory to rewrite the model and perform the sampling analytically to find the exact posterior distribution:
\begin{enumerate}
\item Draw $\theta=(\lambda,p,\alpha) \sim \pi(\cdot)$
\item Simulate a tree $\mathcal{T}$ using parameter $\lambda$ and  count $N$
\item Accept $\theta$ with probability $\prod_{i=1}^{k}\binom{N_i}{D_i} \alpha^{D_i} (1-\alpha)^{N_i-D_i}$.
\end{enumerate}
This algorithm gives exact draws from the posterior distribution of $\theta$ given $D$, and in theory there is no need for any assumption of measurement error.
Note that $\theta$ can include  parameter $\alpha$ for the sampling rate, to be inferred along with the other model parameters. However, this  makes finding a normalizing constant in step 3 difficult. Without a normalizing constant to increase the acceptance rate, applying this algorithm directly will be slow for many values of $D$ and $k$ (the choice of prior distribution and number of parameters we choose to include in  $\theta$ can also have a significant effect on the efficiency). A practical solution would be to add an error term and assume the presence of measurement error on the data (which is likely to exist in this case), in order to increase the acceptance probability in step 3. Approaching the problem in this way, it is possible to carefully model the error on $D$ and improve the estimate of the divergence time.
\end{example}

Using $\epsilon$ to represent measurement error is straight forward, whereas using $\epsilon$ to  model the  model discrepancy (to account for the fact the model is wrong) is  harder to conceptualize and not as commonly used. For deterministic models, the idea of using a model error term when doing calibration or data assimilation is well established and described for a Bayesian framework in \citet{Kennedy_etal01}. They assume that the model run at its `best' input, $\eta(\hat{\theta})$, is sufficient for the model when estimating $\hat{\theta}$. In other words, knowledge of  the model run at its best input  provides all the available  information about the system for the purpose of prediction. If this is the case, then we can  define $\epsilon$ to be the difference between $\eta(\hat{\theta})$ and $D$, and assume $\epsilon$ is independent of $\eta(\hat{\theta})$. Note that the error is the difference between the data and the model run at its best input, and does not  depend on $\theta$. If we do not include an error term $\epsilon$, then  the best input is the value of $\theta$ that best explains the data, given the model.  
When we include an error term which is  carefully modelled and represents our  beliefs about the discrepancy between $\eta(\cdot)$ and reality, then it can be argued that $\hat{\theta}$ represents the `true' value of $\theta$, and that $\pi(\hat{\theta}\mid D, \epsilon \sim \pi_\epsilon(\cdot))$ should be our posterior distribution for $\hat{\theta}$ in light of the data  and the model.

For  deterministic models, \citet{Goldstein_etal08} provide a framework to help think about the model discrepancy. To specify the distribution of $\epsilon$, it can help to break the discrepancy down into various parts: physical  processes not modelled, errors in the specification of the model, imperfect implementation etc. So for example, if $\eta(\cdot)$ represents a global climate model predicting average temperatures, then common model errors could be not including processes such as clouds, CO$_2$ emissions from vegetation etc., error in the specification might be using an unduly simple  model of economic activity, and imperfect implementation would include using grid cells too large to accurately solve the underlying differential equations.  
In some cases it may be necessary to consider model and measurement error, $\epsilon +e$ say, and model each process separately. For stochastic models, as far as we are aware, no guidance exists about how to model the error, and indeed it is not clear what $\epsilon$ should represent.

To complicate matters further, for many models the dimension of $D$ and $X$ will be large, making it likely that the acceptance rate  $\pi_\epsilon(X-D)$ will be small. As noted above, in this case it is necessary to summarize the model output and the data using a  multidimensional summary $S(\cdot)$. Using a summary
means that rather than approximating $\pi(\theta\mid D)$, the algorithms approximate $\pi(\theta\mid S(D))$. The interpretation of $\epsilon$ as model or measurement error still holds, but now the error is on the measurement $S(D)$ or the model prediction $S(X)$. If each element of $S(\cdot)$ has an interpretation in terms of a physical process, this may make it easier to break the error down into independent components. For example, suppose that we use $S(x)=(\bar{x}, s_{xx})$, the sample mean and variance of $X$, and that we then use the following acceptance density 
\[
\pi_\epsilon(S(X)-S(D))=\pi_1(\bar{X}-\bar{D})\pi_2(s_{XX}-s_{DD}).
\]
This is equivalent to assuming that there are two sources of model error. Firstly,  the mean prediction is assumed  to be wrong, with the error distributed with density $\pi_1(\cdot)$. Secondly, it assumes that the model prediction of the  variance is wrong, with the error having distribution $\pi_2(\cdot)$. It also   assumes  that the error in the mean prediction is independent of the error in the variance prediction. This independence is not necessary, but  helps with visualization and elicitation. For this reason it can be helpful to choose the different components of $S(\cdot)$ so that they are close to  independent (independence may also help increase the acceptance rate). Another possibility for choosing $S(\cdot)$ is to use principal component analysis (if $\dim(X)$ is large) to find a smaller number of uncorrelated summaries of the data which may  have meaningful interpretations. In general however, it is not known how to choose good summaries.  \citet{Joyce_etal08} have suggested a method for selecting between different summaries and for deciding how many summaries it is optimal to include. However, more work is required to find summaries which are informative,  interpretable  and for which we can describe the model error.

Finally,  once we have specified a distribution for $\epsilon$, we may find the acceptance rate is too small to be practicable and that it is necessary to compromise (as in Example~\ref{example:BP} above). A pragmatic way to increase the acceptance rate is to use a more disperse  distribution for $\epsilon$. This moves us from the realm of using $\epsilon$ to model an error we believe exists, to using it to approximate the true posterior. This is currently how most  ABC methods are  used. However,  even when making a pragmatic compromise,  the interpretation of the approximation in terms of an error will allow us to think more carefully about how to choose between different compromise solutions.

\begin{example}\label{example:demog}
One of the first uses of an ABC algorithm was by \citet{Pritchard_etal99}, who used a simple stochastic model to study the demographic history of the $Y$ chromosome, and used an approximate rejection algorithm to infer mutation and demographic parameters for their model. Their data consisted of 445 Y chromosomes sampled at eight different loci from a mixture of populations from around the world, which they summarized by just three statistics: the mean (across loci) of the variance of repeat  numbers $V$, the mean effective heterozygosity $H$, and the number of distinct haplotypes $N$. The observed value of the summaries  for their sample was $D\equiv (V, H, N)^T=(1.149, 0.6358, 316)^T$.  They elicited prior distributions for the mutation rates from the literature, and used diffuse priors for population parameters such as the growth rate and the effective number of ancestral Y chromosomes.  Population growth was modelled using a standard coalescent model growing at an exponential rate from a constant ancestral level, and various different mutation models were used to simulate sample values for the three summaries measured in the data. They then applied Algorithm A using the metric 
\begin{equation}
\rho(D,X)=\prod_{i=1}^3 \frac{D_i-X_i}{D_i}
\label{eqn:Pritchard}
\end{equation}
where $X$ is a triplet of simulated values for the three summaries statistics. They used a tolerance value of $\delta=0.1$, which for their choice of metric corresponds to an error of 10\% on each measurement. This gives results equivalent to assuming that there is independent uniform measurement error on the three data summaries, so that the true values of the three summaries have the following distributions
\[
 V\sim\mathcal{U}[1.0341, 1.2624], \;H\sim \mathcal{U}[0.58122, 0.71038], \;N\sim\mathcal{U}[284, 348].
\]
\citet{Beaumont_etal02} used the same model and data set to compare the relative performance of Algorithm A with an algorithm similar to Algorithm B, using an Epanechnikov density applied to the metric value (\ref{eqn:Pritchard}) for the acceptance probability $\pi_\epsilon(\cdot)$. They set a value of $\delta$ (the cut-off in Algorithm A and the range of the support for $\epsilon$ in Algorithm B) by using a quantile $P_\delta$ of the empirical distribution function of simulated values of $\rho(D,X)$, i.e., $P_{0.01}$ means they accepted the 1\% of model runs with values closest to $D$. They concluded that Algorithm B gives more accurate results than Algorithm A, meaning  that the  distribution found using Algorithm B is closer to the posterior  found when assuming no measurement error ($\delta=0$). 

The conclusion that Algorithm B is preferable to Algorithm A for this model is perhaps not surprising in light of what we now know, as it was not taken into account that both algorithms used the same value of $\delta$. For Algorithm A this corresponds to assuming a measurement error with variance $\delta^2/3$, whereas using an Epanechnikov acceptance probability is equivalent to assuming a measurement error with variance $\delta^2/5$. Therefore, using Algorithm B uses measurement error only 60\% as variable as that assumed in Algorithm A, and so it is perhaps not surprising that Algorithm B gives more accurate results in this case.

\end{example}

\section{Approximate Markov chain Monte Carlo}

For problems which have a tightly constrained posterior distribution (relative to the prior), repeatedly drawing independent values of $\theta$ from its prior distribution in the rejection algorithm can be inefficient. For problems with a high dimensional $\theta$ this inefficiency is likely to make the application of a rejection type algorithm impracticable. The idea behind Markov chain Monte Carlo (MCMC) is to build a Markov chain on $\theta$ and correlate successive observations so that more time is spent in regions of high posterior probability. Most MCMC algorithms, such as the Metropolis-Hastings algorithm, depend on knowledge of the likelihood function which we have assumed is not known.  \citet{Marjoram_etal03}  give an approximate version of the Metropolis-Hastings algorithm, which approximates the acceptance probability by using simulated model output with a metric and a 0-1 cut-off to approximate the likelihood ratio. This, as before, is equivalent to assuming  uniform error on a set defined by the metric and  the tolerance. As above, this algorithm can be generalized from assuming uniform measurement error to an arbitrary error term. Below, are two different algorithms to perform MCMC for the model described by Equation~(\ref{eqn:disc_model}). The difference between the two algorithms  is in the choice of sample space used  to construct the stationary distribution. In Algorithm C we consider the state variable to belong to the space of parameter values $\Theta$, and construct a  Markov chain  $\{\theta_1, \theta_2, \ldots\}$ which obeys the following dynamics:
\begin{enumerate}
\item[]{\bf Algorithm C: probabilistic approximate MCMC 1}
\item[C1] At time $t$, propose a move from $\theta_t$ to $\theta'$ according to transition kernel $q(\theta_t, \theta')$.
\item[C2] Simulate $X'\sim \eta(\theta')$.
\item[C3] Set $\theta_{t+1}=\theta'$ with probability 
\begin{equation}\label{eqn:C_accep_prob}
r(\theta_t, \theta'\mid X')=\frac{\pi_\epsilon(D-X')}{c} \min\left(1, \frac{q(\theta', \theta_t) \pi(\theta')}{q(\theta_t, \theta')\pi(\theta_t)}\right),
\end{equation}
 otherwise set $\theta_{t+1}=\theta_t$.
\end{enumerate}
An alternative approach is to introduce the value of the simulated output as an auxiliary variable and construct the Markov chain on the space $\Theta \times \mathcal{X}$, where $\mathcal{X}$ is the space of model outputs.

\begin{enumerate}
\item[]{\bf Algorithm D: probabilistic approximate MCMC 2}
\item[D1] At time $t$, propose a move from $\psi_t=(\theta_t, X_t)$ to $\psi'=(\theta', X')$ with $\theta'$ drawn from transition kernel $q(\theta_t, \theta')$, and $X'$ simulated from the model using $\theta'$:
 \[
 X'\sim \eta(\theta')
 \]
\item[D2] Set $\psi_{t+1}=(\theta', X')$ with probability 
\begin{equation}\label{eqn:D_accep_prob}
r((\theta_t, X_t), (\theta', X'))=\min\left(1,  \frac{\pi_\epsilon(D-X')q(\theta', \theta_t) \pi(\theta')}{\pi_\epsilon(D-X_t)q(\theta_t,\theta')\pi(\theta_t)}\right),
\end{equation}
 otherwise set $\psi_{t+1}=\psi_t$.
\end{enumerate}

\begin{proof}[of convergence]
To show that these Markov chains  converge to the required posterior distribution, it is sufficient to show that the chains satisfy the detailed balance equations
\begin{equation*}
\pi(s)p(s,t)=\pi(t)p(t,s) \quad\mbox{for all } s,t
\end{equation*} where $p(\cdot, \cdot)$ is the transition kernel of the chain and $\pi(\cdot)$ the required stationary distribution.

For Algorithm C the transition kernel is the product of $q(\theta, \theta')$ and the acceptance rate. To calculate the acceptance rate, note that in Equation~(\ref{eqn:C_accep_prob}) the acceptance probability is conditioned upon knowledge of $X'$ and so we must integrate out $X'$  to find $r(\theta, \theta')$. This gives the transition kernel for the chain:
\begin{equation*}
p(\theta_t, \theta')=q(\theta_t, \theta') \int_\mathcal{X} \frac{\pi_\epsilon(D-X')}{c} \min\left(1, \frac{q(\theta', \theta_t) \pi(\theta')}{q(\theta_t, \theta')\pi(\theta_t)}\right) \pi(X'\mid \theta') \d X'.
\end{equation*}
The target stationary distribution is 
\[
\pi(\theta\mid D)=\frac{\pi(\theta) \int_\mathcal{X} \pi_\epsilon(D-X) \pi(X\mid \theta) \d X}{\pi(D)}.
\]
It is then simple to show that the Markov chain described by Algorithm C satisfies the detailed balance equations (see \citet{Liu_01} for comparable calculations).

For Algorithm D, the transition kernel is 
\begin{equation}\label{eqn:D1}
p((\theta_t, X_t), (\theta', X'))=q(\theta_t, \theta') \pi(X' \mid \theta') \min\left(1,  \frac{\pi_\epsilon(D-X')q(\theta', \theta_t) \pi(\theta')}{\pi_\epsilon(D-X_t)q(\theta_t,\theta')\pi(\theta_t)}\right).
\end{equation}
The Markov chain in this case takes values in $\Theta \times \mathcal{X}$ and the required stationary distribution is
\begin{equation}\label{eqn:D2}
\pi(\theta, X \mid D)=\frac{\pi_\epsilon(D-X)\pi(X\mid \theta)\pi(\theta)}{\pi(D)}.
\end{equation}
It can then be shown that Equations~(\ref{eqn:D1}) and~(\ref{eqn:D2}) satisfy the detailed balance equations. 
\end{proof}
While Algorithm C is more recognisable as a generalization of the approximate MCMC algorithm given in \citet{Marjoram_etal03}, Algorithm D is likely to be more efficient in most cases. This is because the ratio of model error densities that occurs in acceptance rate~(\ref{eqn:D_accep_prob}) is likely to result in larger probabilities than those given by Equation~(\ref{eqn:C_accep_prob}) which simply has a $\pi_\epsilon(D-x)/c$ term instead. Algorithm D also has the advantage of not requiring a normalizing constant.

\section{Extensions}

\subsection{ Importance sampling}

Suppose our aim is to calculate expectations  of the form
\[
\BE(f(\theta)\mid D)=\int f(\theta) \pi(\theta\mid D) \d \theta
\]
where the expectation is taken with respect to the posterior distribution of $\theta$. The simplest way to approximate this is to draw a sample of $\theta$ values $\{\theta_i\}_{i=1, \ldots, n}$, from $\pi(\theta\mid D)$ using Algorithm B, C or D and then approximate using the sum $n^{-1} \sum f(\theta_i)$. However, a more stable estimator can be obtained by using draws from the  prior weighted by $\pi_\epsilon(D-X_i)$ as in Algorithm B. For each  $(\theta_i, X_i)$ pair drawn from the prior and simulator in steps B1 and B2, assign $\theta_i$ weight $w_i=\pi_\epsilon(D-X_i)$. Then an estimator of the required expectation is 
\[
\frac{\sum f(\theta_i)w_i}{\sum w_i}.
\]
This is an importance sampling algorithm targeting the joint distribution
$$\pi(X, \theta \mid D) \propto \pi_\epsilon (D-X) \pi(X\mid\theta) \pi(\theta)$$
using instrumental distribution $\pi(X\mid \theta) \pi(\theta)$. Note that if the uniform acceptance kernel 
$$\pi_\epsilon(D-X) \propto \mathbb{I}_{\rho(D,X)\leq \delta}$$
is used, then this approach reduces to the rejection algorithm, as proposals are given weight 1 (accepted) or 0 (rejected), showing that there is no uniform-ABC version of importance sampling.

Sequential Monte Carlo algorithms are possible however, and \citet{Sisson_etal07} and \citet{Toni_etal09} considered algorithms in which the tolerance $\delta$ is slowly reduced through a schedule $\delta_1, \ldots, \delta_T$ to some small final value. Both of these algorithms, as well as variants such as \citet{Drovandi_etal2011} and \citet{DelMoral_etal2012}, which use Metropolis-Hastings moves of the parameter between iterations to provide more efficient proposals, can be extended to the generalised ABC case using general acceptance kernels. The move from 0-1 cut-offs to general acceptance rates can introduce difficulties with memory constraints, due to the requirement to store a large number of particles, many with small but non-zero weights. Partial rejection control, introduced by \citet{Liu_01} and extended to an ABC setting by \citet{Peters_etal2012}, can be used to reject particles with small weights, only keeping particles which have weight above some threshold.



\subsection{Model selection}

The theoretical acceptance rate from the rejection algorithm (Algorithm A with $\delta=0$) is equal to the model evidence $\pi(D)$. The evidence from different models can then be used to calculate Bayes factors  which can be used to perform model selection \citep{Kass_etal95}. It is possible to approximate the value of $\pi(D)$ by using the acceptance rate from Algorithm A \citep{Wilkinson_07}. By doing this for two or more competing models,  we can perform approximate model selection, although in practice this approach can be unstable as $\delta$ varies \citep{Wilkinson_07}. The estimate of $\pi(D)$ can be improved and made interpretable by using the weighted estimate
\[
\frac{1}{n}\sum_{i=1}^n \frac{1}{m}\sum_{j=1}^m \pi_\epsilon(D-X_i^j)
\]
where $X_i^1, \ldots, X_i^m\sim \eta(\theta_i)$ and $\theta_1, \ldots, \theta_n \sim \pi(\cdot)$. This gives a more stable estimator than simply taking the acceptance rate, and  also tends to the exact value (as $n, m\rightarrow \infty$) for the model given by Equation~(\ref{eqn:disc_model}).

\citet{Robert_etal2011} and \citet{Didelot_etal2011} have highlighted the dangers of using ABC algorithms for model selection when $D$ and $X$ are replaced by summaries $S(D)$ and $S(X)$. The Bayes factor based on the full data $D$, will in general differ from the Bayes factor based on the summary $S(D)$, even when $S$ is a sufficient statistic for $\theta$ for both  simulators.

The approach advocated here, of considering ABC as an extension of the modelling process, using acceptance kernel $\pi_\epsilon(D -X)$ (or $\pi_\epsilon(S(D) | S(X))$ in a more general non-additive setting) to represent the relationship between simulator output and observations (as encapsulated by Theorem \ref{thm:1}), suggests a different approach. The choice of acceptance kernel should be made after careful consideration of the simulator's ability, and inevitably involves a degree of subjective judgement (as does the choice of simulator, prior, and summary statistics). The kernel used forms part of the statistical model, and any model selection scheme will assess this choice, as well as the choice of simulator and prior. In other words, it is inevitable that the estimated Bayes factor will depend upon $\pi_\epsilon$ in general, further highlighting the need for its careful design.

Similarly, the choice of summary statistic $S(D)$ used to reduce the dimension of the data and simulator output, should be based on careful consideration of what aspects of the data we expect the simulator to be able to reproduce. Recent work by \citet{Nunes_etal2010, Fearnhead_etal2012, Barnes_etal2012} and \citet{Prangle_etal2013} has focussed on automated methods for choosing good summary statistics, but care should be taken to ensure the summaries selected coincide with the modeller's expectations of what the simulator can reproduce. Examples can be constructed in which summaries are strongly informative about the parameters (in the sense of $\pi(\theta |S(D))$ differing from $\pi(\theta)$), but which do not produce believable posteriors. For example, in dynamical system models, phase sensitive summaries (such as the sum of square differences) are usually informative about the simulator parameters, even though the simulators were only designed to capture the phase-insensitive parts of the system. Using these summaries will give the appearance of having learned about the parameters, as the posterior will differ from the prior, but it is unclear whether what has been learnt is of value.  If the summary $S(D)$ is chosen on a sound physical basis, and the inference viewed as conditional upon this choice (i.e., the posterior is taken to be $\pi(\theta | S(D))$ and is not seen as an approximation to $\pi(\theta |D)$), then the difficulties for ABC model selection raised by \cite{Robert_etal2011} are circumvented, and interpretation is clear.

\section{Discussion}

It has been shown in this  paper  that approximate Bayesian computation algorithms  can be considered to give exact inference under the assumption of model error.
However, this is only part of the way towards a  complete understanding of  ABC algorithms. In the majority of the application papers using ABC methods, summaries of the data and model output have been used to reduce the dimension of the output. It cannot be known whether these summaries are sufficient for the data, and so in most cases the use of summaries means that there is  another layer of approximation. While this work allows us to understand the error assumed on the measurement of the summary, it says nothing about what effect using the summary rather than the complete data has on the inference. 

The use of a simulator discrepancy  term when making inferences is important if one wants to move from making statements about the simulator  to statements about reality. There has currently been only minimal work done on modelling the discrepancy term for stochastic models. One way to approach this is  to view the model as deterministic, outputting a density $\pi_\theta(x)$ for each value of the input $\theta$ (many realizations of $\eta(\theta)$ would be needed to learn $\pi_\theta(x)$).  The discrepancy term $\epsilon$ can then be considered as representing the difference between $\pi_\theta(x)$ and the true variability inherent in the physical system. 


\section*{Acknowledgements}

I would like to thank Professor Paul Blackwell for his suggestion that the metric in the ABC algorithm might have a probabilistic interpretation. I would also like to thank Professor Simon Tavar\'e, Professor Tony O'Hagan, and Dr Jeremy Oakley for  useful discussions on an early draft of the manuscript.

\bibliographystyle{elsarticle-harv}
\bibliography{ABCisExact}

\end{document}